%% file: arxiv_pbwt.tex
\documentclass[11pt,a4paper]{article}

\usepackage[T1]{fontenc}

\usepackage{graphicx}

\input{macro}

\usepackage[top=2cm, bottom=2cm, left=2cm, right=2cm, includefoot]{geometry}
\usepackage{authblk}
\usepackage{amsthm}
\usepackage{url}

\newtheorem{proposition}{Proposition}
\newtheorem{definition}{Definition}
\newtheorem{theorem}{Theorem}
\newtheorem{lemma}{Lemma}
\newtheorem{corollary}{Corollary}

\makeatletter
\renewenvironment{proof}[1][\proofname]{\par
	\normalfont \topsep6\p@\@plus6\p@\relax
	\trivlist
	\item\relax
	{\itshape
		#1\@addpunct{.}}\hspace\labelsep\ignorespaces
}{%
	\popQED\endtrivlist\@endpefalse
}
\makeatother

\begin{document}

\title{Computing the Parameterized Burrows--Wheeler Transform Online}

\author[1]{Daiki~Hashimoto}
\author[1]{Diptarama~Hendrian}
\author[2]{Dominik~K\"{o}ppl}
\author[1]{Ryo~Yoshinaka}
\author[1]{Ayumi~Shinohara}

\date{}

\affil[1]{Tohoku University, Japan}
\affil[2]{Tokyo Medical and Dental University, Japan}

\maketitle

\begin{abstract}
Parameterized strings are a generalization of strings in that their characters are drawn from two different alphabets,
where one is considered to be the alphabet of static characters and the other to be the alphabet of parameter characters.
Two parameterized strings are a parameterized match if there is a bijection over all characters such that the bijection transforms one string to the other while keeping the static characters (i.e., it behaves as the identity on the static alphabet).
Ganguly et al.~[SODA 2017] proposed the parameterized Burrows--Wheeler transform (pBWT) as a variant of 
the Burrows--Wheeler transform for space-efficient parameterized pattern matching.
In this paper, we propose an algorithm for computing the pBWT online by reading the characters of a given input string one-by-one from right to left.
Our algorithm works in $O(|\Pi| \log n / \log \log n)$ amortized time for each input character,
where $n$ and $\Pi$ denote the size of the input string and the alphabet of the parameter characters, respectively.

\date{}
\end{abstract}

\input{intro}
\input{preliminaries}

\input{algorithm}

\bibliographystyle{plain}
\bibliography{ref}

\appendix
\input{appendix}

\end{document}

%% file: macro.tex
\usepackage{tikz}
\usepackage{amsmath,amssymb,enumerate,mathrsfs}
\usepackage[ruled,linesnumbered,algo2e,vlined]{algorithm2e}
\usepackage{cleveref}
\usepackage{stmaryrd}
\usepackage{enumitem}
\usepackage{arydshln} 
\usepackage{thmtools}
\usepackage{thm-restate}

\newif\ifcameraready

\camerareadytrue

\newcommand{\scl}{0.85}

\newcommand{\msf}[1]{\mathsf{#1}}
\newcommand{\mbb}[1]{\mathbb{#1}}
\newcommand{\mtt}[1]{\mathtt{#1}}

\newcommand{\rest}[2]{#1 \mathbin{\upharpoonright} #2}

\newcommand{\inft}{\text{\scalebox{0.7}[1]{\boldmath$\infty$}}}
\newcommand{\ift}{\text{\scalebox{0.5}[1]{\boldmath$\infty$}}}

\newcommand{\tta}{\mathtt{a}}
\newcommand{\CL}{\mathit{Left}}
\newcommand{\CR}{\mathit{Right}}
\newcommand{\SFL}{\mathsf{Left}}
\newcommand{\SFR}{\mathsf{Right}}
\newcommand{\bdl}{\$}
\newcommand{\rd}[1]{\underline{\textcolor{red}{#1}}}

\newcommand{\llrrbracket}[1]{\llbracket #1 \rrbracket}

\newcommand{\RA}{\mathit{RA}}
\newcommand{\LF}{\mathit{LF}}
\newcommand{\Rot}{\mathit{Rot}}
\newcommand{\PBWTF}{\mathit{pBWT}}
\newcommand{\PBWT}{\msf{L}}
\newcommand{\First}{\msf{F}}
\newcommand{\LCP}{\msf{LCP}^{\inft}}
\newcommand{\Lcp}{\mathit{lcp}^{\inft}}
\newcommand{\RM}{\msf{RM}}
\newcommand{\SFC}{\msf{C}}

\newcommand{\Rank}{\mathsf{rank}}
\newcommand{\Select}{\mathsf{select}}
\newcommand{\Access}{\mathsf{access}}
\newcommand{\Insert}{\mathsf{insert}}
\newcommand{\Delete}{\mathsf{delete}}

\newcommand{\Cnt}{\mathit{cnt}}
\newcommand{\Dol}{k_S}
\newcommand{\Pos}{k_T}

\newcommand{\lrangle}[1]{\langle #1 \rangle}

\newcommand{\UpdateLF}{\mathsf{UpdateLF}}
\newcommand{\InsertDol}{\mathsf{InsertRow}}
\newcommand{\UpdateLCP}{\mathsf{UpdateLCP}}
\newcommand{\UpdateAll}{\mathsf{UpdateAll}}

\SetKwProg{Fn}{Function}{}{end}
\SetArgSty{textup} 
\SetCommentSty{mycommfont}
\SetInd{6pt}{6pt}

%% file: intro.tex
\section{Introduction}
The \emph{parameterized matching problem} (\emph{p-matching problem})~\cite{PMA}
is a generalization of the classic pattern matching problem in the sense that
we here consider two disjoint alphabets, the set $\Sigma$ of \emph{static characters} and the set $\Pi$ of \emph{parameter characters}.
We call a string over $\Sigma \cup \Pi$ a \emph{parameterized string (p-string)}.
Two equal-length p-strings $X$ and $Y$ 
are said to \emph{parameterized match} (\emph{p-match})
if there is a bijection that renames the parameter characters in $X$ so $X$ becomes equal to $Y$.
The \emph{p-matching problem} is, given a text p-string $T$ and pattern p-string $P$,
to output the positions of all substrings of $T$ that p-match $P$.
The p-matching problem is motivated
by applications in the software maintenance~\cite{Baker93,PMA},
the plagiarism detection~\cite{EPM},
the analysis of gene structures~\cite{Shibuya04}, and so on.
There exist indexing structures that support p-matching,
such as parameterized suffix trees~\cite{Baker93,Shibuya04}, parameterized suffix arrays~\cite{FujisatoSA2019,TomohiroDBIT09}, and so on~\cite{Diptarama2017,Fujisato2018,Nakashima2020,Nakashima2020plst}; see also~\cite{DBLP:journals/dam/MendivelsoTP20} for a survey.
A drawback of these indexing structures is that they have high space requirements.

A more space-efficient indexing structure, the \emph{parameterized Burrows--Wheeler transform (pBWT)},
was proposed by Ganguly et al.~\cite{Ganguly2017}.
The pBWT is a variant of the \emph{Burrows--Wheeler transform (BWT)}~\cite{Burrows94ablock-sorting} that can be used as an indexing structure for p-matching using only $o(n \log n)$ bits of space.
Later on, Kim and Cho~\cite{KimC21} improved this indexing structure by changing the encoding of p-strings used for defining the pBWT\@.
Recently, Ganguly et al.~\cite{0002ST22} augmented this index with capabilities of a suffix tree while keeping the space within $o(n \log n)$ bits.
However, as far as we are aware of, none research related to the pBWT~\cite{Ganguly2017,KimC21,0002ST22,Thankachan22} has discussed how to construct their pBWT-based data structures in detail.
Their construction algorithms mainly rely on the parameterized suffix tree.
Given the parameterized suffix tree of a p-string $T$ of length $n$,
the pBWT of $T$ can be constructed in $O(n\log(|\Sigma| + |\Pi|))$ time offline.

In this paper, we propose an algorithm for constructing pBWTs and related data structures used for indexing structures of p-matching.
Our algorithm constructs the data structures directly in an online manner by reading the input text from right to left.
The algorithm uses the dynamic array data structures of Navarro and Nekrich~\cite{Navarro2014} to maintain our growing arrays.
For each character read, our algorithm takes $O(|\Pi| \log n / \log \log n)$ amortized time,
where $n$ is the size of input string.
Therefore, we can compute pBWT of a p-string $T$ of length $n$ in $O(n|\Pi| \log n / \log \log n)$ time in total.
In comparison, computing the standard BWT on a string~$T$ (i.e., the pBWT on a string having no parameter characters)
can be done in $O(n \log n / \log \log n)$ time with the dynamic array data structures~\cite{Navarro2014} (see~\cite{policriti15fast} for a description of this online algorithm).
Looking at our time complexity, the factor $|\Pi|$ also appears in the time complexity of an offline construction algorithm of parameterized suffix arrays~\cite{FujisatoSA2019} as $O(n|\Pi|)$ and a right-to-left online construction algorithm of parameterized suffix trees~\cite{Nakashima2020} as $O(n|\Pi|\log(|\Pi| + |\Sigma|))$.
This suggests it would be rather hard to improve the time complexity of the online construction of pBWT to be independent of $|\Pi|$.

%% file: preliminaries.tex
\section{Preliminaries}
We denote the set of nonnegative integers by $\mbb{N}$ and let $\mbb{N}_+ = \mbb{N} \setminus \{0\}$ and $\mbb{N}_\inft = \mbb{N}_+ \cup \{\inft\}$.
The set of strings over an alphabet $A$ is denoted by $A^*$.
The empty string is denoted by $\varepsilon$.
The length of a string $W \in A^*$ is denoted by $|W|$.
For a subset $B \subseteq A$, the set of elements of $B$ occurring in $W \in A^*$ is denoted by $\rest{B}{W}$.
We count the number of occurrences of characters of $B$ in a string $W$ by $|W|_B$.
So, $|W|_A = |W|$.
When $B$ is a singleton of $b$, i.e., $B = \{b\}$, we often write $|W|_{b}$ instead of $|W|_{\{b\}}$.
When $W$ is written as $W = XYZ$, $X$, $Y$, and $Z$ are called 
\emph{prefix}, \emph{factor}, and \emph{suffix} of $W$, respectively.
The $i$-th character of $W$ is denoted by $W[i]$ for $1 \leq i \leq |W|$.
The factor of $W$ that begins at position $i$ and ends at position $j$ 
is $W[i:j]$ for $1 \leq i \leq j \leq |W|$.
For convenience,
we abbreviate $W[1:i]$ to $W[{}:i]$ and $W[i:|W|]$ to $W[i:{}]$ for $1 \leq i \leq |W|$.
Let  $\Rot(W,0) = W$ and $\Rot(W,i+1) = \Rot(W,i)[|W|]\Rot(W,i)[:|W|-1]$ be the $i$-th right rotation of $W$.
Note that $\Rot(W,i) = \Rot(W, i +|W|)$.
For convenience we denote $W_i = \Rot(W,i)$.
Let $\CL_W(a)$ and $\CR_W(a)$ be the leftmost and rightmost positions of a character $a \in A$ in $W$, respectively.
If $a$ does not occur in $W$, define $\CL_W(a) = \CR_W(a) = 0$.

\subsection{Parameterized Burrows--Wheeler transform}

Throughout this paper, we fix two disjoint ordered alphabets $\Sigma$ and $\Pi$.
We call elements of $\Sigma$ \emph{static characters}
and those of $\Pi$ \emph{parameter characters}.
Elements of $\Sigma^{*}$ and $(\Sigma \cup \Pi)^{*}$ are called \emph{static strings} and \emph{parameterized strings} (or \emph{p-strings} for short), respectively.

Two p-strings $S$ and $T$ of the same length are a \emph{parameterized match} \emph{(p-match)},
 denoted by $S \approx T$,
if there is a bijection $f$ on $\Sigma \cup \Pi$ such that $f(a) = a$ for any $a \in \Sigma$ and $f(S[i]) = T[i]$ for all $1 \leq i \leq |T|$~\cite{PMA}.
We use Kim and Cho's version of p-string encoding~\cite{KimC21}, which replaces $0$ in Baker's encoding~\cite{Baker93} by $\inft$.
The \emph{prev-encoding} $\lrangle{T}$ of $T$ is the string over $\Sigma \cup \mbb{N}_\inft$ of length $|T|$ defined by
\begin{align*}
	\lrangle{T}[i] = 
	\begin{cases}
		T[i] & \text{if } T[i] \in \Sigma , \\
		\inft          & \text{if $T[i] \in \Pi$ and $\CR_{T[:i-1]}(T[i]) = 0$,} \\
		i-\CR_{T[:i-1]}(T[i])  		& \text{if  $T[i] \in \Pi$ and $\CR_{T[:i-1]}(T[i]) \ne 0$} 
	\end{cases}
\end{align*}
for $1 \le i \le |T|$.
When $T[i] \in \Pi$, $\lrangle{T}[i]$ represents the distance between $i$ and the previous occurrence position of the same parameter character.
If $T[i]$ does not occur before the position $i$, the distance is assumed to be $\inft$.
We call a string $W \in (\Sigma \cup \mbb{N}_\inft)^*$ a \emph{pv-string} if $W = \lrangle{T}$ for some p-string $T$.
For any p-strings $S$ and $T$,  $S \approx T$ if and only if $\lrangle{S}=\lrangle{T}$~\cite{PMA}.
For example, given $\Sigma = \{\mathtt{a}, \mathtt{b}\}$ and $\Pi = \{\mathtt{u}, \mathtt{v}, \mathtt{x}, \mathtt{y}\}$,
$S = \mathtt{uvvauvb}$ and $T = \mathtt{xyyaxyb}$ are a p-match 
by $f$ with $f(\mathtt{u})=\mathtt{x}$ and $f(\mathtt{v})=\mathtt{y}$,
where $\lrangle{S} = \lrangle{T} = \inft \inft 1\mathtt{a}43\mathtt{b}$.

For defining pBWT, we use another encoding $\llrrbracket{T}$ given by
\begin{align*}
	\llrrbracket{T}[i] = 
	\begin{cases}
		T[i] & \text{if } T[i] \in \Sigma , \\
		| \rest{\Pi}{T_{n-i}[1:\CL_{T_{n-i}}(T[i])]} |  		& \text{if } T[i] \in \Pi
	\end{cases}
\end{align*}
for $1 \le i \le |T|$.
When $T[i] \in \Pi$, $\llrrbracket{T}[i]$ counts the number of distinct parameter characters in $T$ between $i$ and the next occurrence of $T[i]$, if $T[i]$ occurs after $i$.
If $i$ is the rightmost occurrence position of $T[i]$, then we continue counting parameter characters from the left end to the right until we find $T[i]$.
Since $T[i]$ occurs in $T_{n-i}$ as the last character, $\llrrbracket{T}[i]$ cannot be zero.
Note that $\llrrbracket{\Rot(T,i)} = \Rot(\llrrbracket{T},i)$ by definition.
It is not hard to see that for any p-strings $S$ and $T$,  $S \approx T$ if and only if $\llrrbracket{S}=\llrrbracket{T}$ (see \Cref{lem:llrrbracketEquiv} in the appendix).
For example, the two strings $S$ and $T$ given above are encoded as $\llrrbracket{S} = \llrrbracket{T} = 212\mathtt{a}22\mtt{b}$.

Hereafter in this section, we fix a p-string $T$ of length $n$ which ends with a special static character ${\$}$ which occurs nowhere else in $T$.
We extend the linear order over $\Sigma$ to $\Sigma \cup \mbb{N}_\inft$ by letting $\$ < a < i < \inft$ for any $a \in \Sigma \setminus\{\$\}$ and $i \in \mbb{N}_+$.
The order over $\mbb{N}_+$ coincides with the usual numerical order. 

The pBWT of $T$ is defined through sorting $\llrrbracket{T_p}$ for $p=1,\dots,n$ using $\lrangle{T_p}$ as keys. 
\begin{definition}[Parameterized rotation array]
	\emph{The parameterized rotation array} $\RA_T$ of $T$ is an array of size $n$ such that $\RA_T[i]=p$ with $1 \le p \le n$ if and only if 
	$\lrangle{T_p}$ is the $i$-th lexicographically smallest string in $\{\,\lrangle{T_p} \mid 1 \le p \le n\,\}$.
	We denote its inverse by $\RA_T^{-1}$, i.e., $\RA_T^{-1}[p]=i$ iff $\RA_T[i]=p$.
\end{definition}
Note that $\RA_T$ and $\RA^{-1}_T$ are well-defined and bijective due to the presence of ${\$}$ in $T$.
Here, we have $\RA_T[i] =  n - \mathit{pSA}_T[i] + 1$,
where $\mathit{pSA}_T$ refers to the suffix array $\mathtt{pSA}_\inft$ in~\cite{KimC21}. 
The array gives an $n \times n$ square matrix $(\llrrbracket{T_{\RA_T[i]}})_{i=1}^n$, which we call the \emph{rotation sort matrix} of $T$, whose $(i,p)$ entry is $\llrrbracket{T_{\RA_T[i]}}[p]$.
The pBWT of $T$ is formed by the characters in the last column of the matrix.
\begin{definition}[pBWT~\cite{KimC21}]
	\emph{The parameterized Burrows--Wheeler transform (pBWT)} of a p-string $T$, denoted by $\PBWTF(T)$, is a string of length $n$ such that $\PBWTF(T)[i] = \llrrbracket{T_{\RA[i]}}[n]$.
\end{definition}
An example pBWT can be found in \Cref{table:PBWT}.
We will use $\PBWT_T$ as a synonym of $\PBWTF(T)$, since it represents the \emph{last} column of the matrix $(\llrrbracket{T_{\RA_T[i]}})_{i=1}^n$.
When picking up the characters from the \emph{first} column, we obtain another array $\First_T$.
That is,  $\First_T[i]=\llrrbracket{T_{\RA_T[i]}}[1]$ for all $i \in \{1,\dots,n\}$.
Those arrays $\PBWT_T$ and $\First_T$ are ``linked'' by the following mapping.
\begin{definition}[LF mapping]
	\emph{The LF mapping} $\mathsf{\LF}_T: \{1,\dots,n\} \rightarrow \{1,\dots,n\}$ for $T$ is defined as $\mathsf{\LF}_T(i) = j$ if $T_{\RA[i]+1} = T_{\RA[j]}$.
\end{definition}
By rotating $T_p$ to the right by one, the last character moves to the first position in $T_{p+1}$.
Roughly speaking, $\PBWT_T[i]$ and $\First_T[\LF_T(i)]$ ``originate'' in the same character occurrence of $T$, which implies $\PBWT_T[i]=\First_T[\LF_T(i)]$ in particular.
One can recover $\llrrbracket{T}$ as $\llrrbracket{T}[p]=\PBWT_T[\LF^{-p}(\Pos)]]$ for $1 \le p \le n$ where $\Pos = \RA_T^{-1}[n]$.
$\PBWT_T$, $\First_T$, and $\LF_T$ are used for pattern matching based on pBWT.
See~\cite{KimC21} for the details.

\begin{table}[t]
	\caption{The pBWT $\PBWTF(T) = \PBWT_T = \mtt{a}33131\$22\mtt{aa}$ of the example string $T=\mathtt{xayzzazyza\$}$ with related arrays, where $\Sigma = \{\mtt{a}\}$ and $\Pi = \{\mtt{x},\mtt{y},\mtt{z}\}$.}
	\label{table:PBWT}
	\centering
	\scalebox{\scl}{
	\begin{tabular}{|c||c|c||c|c|c|c|c|c|}
		\hline
		$i$ & $T_i$ & $\lrangle{T_i}$ & $\RA_T[i]$ & $\LCP_T[i]$ & $\lrangle{T_{\RA_T[i]}}$ & $\First_T[i]$ & $\llrrbracket{T_{\RA_T[i]}}$ & $\PBWT_T[i]$ \\
		\hline
		  1 & $\mathtt{\$xayzzazyza}$ & $\$\ift\tta\ift\ift1\tta252\tta$ &  1 & 0 & $\$\ift\tta\ift\ift1\tta252\tta$ & $\$$   & $\$3\tta211\tta233\tta$ & $\tta$ \\
		  2 & $\mathtt{a\$xayzzazyz}$ & $\tta\$\ift\tta\ift\ift1\tta252$ &  2 & 0 & $\tta\$\ift\tta\ift\ift1\tta252$ & $\tta$ & $\tta\$3\tta211\tta233$ & $3$    \\
		  3 & $\mathtt{za\$xayzzazy}$ & $\ift\tta\$\ift\tta\ift61\tta25$ & 10 & 2 & $\tta\ift\ift1\tta252\tta\$\ift$ & $\tta$ & $\tta211\tta233\tta\$3$ & $3$    \\
		  4 & $\mathtt{yza\$xayzzaz}$ & $\ift\ift\tta\$\ift\tta661\tta2$ &  6 & 0 & $\tta\ift\ift2\tta\$\ift\tta661$ & $\tta$ & $\tta233\tta\$3\tta211$ & $1$    \\
		  5 & $\mathtt{zyza\$xayzza}$ & $\ift\ift2\tta\$\ift\tta661\tta$ &  3 & 1 & $\ift\tta\$\ift\tta\ift61\tta25$ & $3$    & $3\tta\$3\tta211\tta23$ & $3$    \\
		  6 & $\mathtt{azyza\$xayzz}$ & $\tta\ift\ift2\tta\$\ift\tta661$ &  7 & 1 & $\ift\tta2\ift2\tta\$\ift\tta66$ & $1$    & $1\tta233\tta\$3\tta21$ & $1$    \\
		  7 & $\mathtt{zazyza\$xayz}$ & $\ift\tta2\ift2\tta\$\ift\tta66$ & 11 & 1 & $\ift\tta\ift\ift1\tta252\tta\$$ & $3$    & $3\tta211\tta233\tta\$$ & $\$$   \\
		  8 & $\mathtt{zzazyza\$xay}$ & $\ift1\tta2\ift2\tta\$\ift\tta6$ &  8 & 1 & $\ift1\tta2\ift2\tta\$\ift\tta6$ & $1$    & $11\tta233\tta\$3\tta2$ & $2$    \\
		  9 & $\mathtt{yzzazyza\$xa}$ & $\ift\ift1\tta252\tta\$\ift\tta$ &  4 & 2 & $\ift\ift\tta\$\ift\tta661\tta2$ & $3$    & $33\tta\$3\tta211\tta2$ & $2$    \\
		 10 & $\mathtt{ayzzazyza\$x}$ & $\tta\ift\ift1\tta252\tta\$\ift$ &  9 & 2 & $\ift\ift1\tta252\tta\$\ift\tta$ & $2$    & $211\tta233\tta\$3\tta$ & $\tta$ \\
		 11 & $\mathtt{xayzzazyza\$}$ & $\ift\tta\ift\ift1\tta252\tta\$$ &  5 & 0 & $\ift\ift2\tta\$\ift\tta661\tta$ & $2$    & $233\tta\$3\tta211\tta$ & $\tta$ \\
		 \hline
	\end{tabular}
	}
\end{table}

Our pBWT construction algorithm maintains neither $\RA_T$ nor $\LF_T$, but involves some helper data structures in addition to $\PBWT_T$ and $\First_T$.
Among those, the array $\LCP_T$ is worth explaining before going into the algorithmic details.
For two pv-strings $X$ and $Y$, let $\Lcp(X,Y) = |W|_\inft$ be the number of $\inft$'s in the longest common prefix $W$ of $X$ and $Y$.
The following array counts the number of $\inft$'s in the longest common prefixes of two adjacent rows in $(\lrangle{T}_{\RA_T[i]})_{i=1}^n$.
\begin{definition}[$\inft$-LCP array]
	\emph{The $\inft$-LCP array} $\LCP_T$ of $T$ is an array of size $n$ such that $\LCP_T[n] = 0$ and $\LCP_T[i] = \Lcp(\lrangle{T_{\RA_T[i]}}, \lrangle{T_{\RA_T[i+1]}})$ for $1 \le i < n$.
\end{definition}
\Cref{table:PBWT} shows an example of a pBWT and related (conceptual) data structures.
We can compute $\Lcp(\lrangle{T_{\RA_T[i]}}, \lrangle{T_{\RA_T[j]}})$ using $\LCP_T$ as follows.
\begin{lemma}\label{lem:lcpmin}
	For $1 \le i < j \le n$,
	$\Lcp(\lrangle{T_{\RA_T[i]}}, \lrangle{T_{\RA_T[j]}}) = \min_{i \le k <j}\LCP_{T}[k]$.
\end{lemma}

Kim and Cho~\cite{KimC21} showed some basic relations among $\PBWT_T$, $\LF_T$, and $\Lcp$.
We rephrase Lemma~3 of~\cite{KimC21} into a form convenient for our discussions.
\begin{lemma}\label{lem:cross}
	Consider $i$ and $j$ with $1 \le i < j \le n$ and $T_{\RA_T[i]}[n] , T_{\RA_T[j]}[n] \in \Pi$.
	Then, $\LF_T(i) < \LF_T(j)$ iff 
	\(
	 \min\{ \PBWT_T[i] - 1, \ \Lcp(\lrangle{T_{\RA_T[i]}},\lrangle{T_{\RA_T[j]}})\} < \PBWT_T[j]
	\).
\end{lemma}
\begin{corollary}[{\cite{KimC21}}]\label{cor:Kim4}
	If $i < j$ and $\PBWT_T[i] = \PBWT_T[j]$, then $\LF_T(i) < \LF_T(j)$.
\end{corollary}

To maintain $\PBWT_T$, $\First_T$, and $\LCP_{T}$ dynamically, our algorithm uses the data structure for dynamic arrays by Navarro and Nekrich~\cite{Navarro2014} that supports the following operations on an array $Q$ of size $m$ in $O(\frac{\log m}{\log \log m})$ amortized time.
\begin{enumerate}
	\item $\Access(Q,i)$: returns $Q[i]$ for $1 \le i \le m$;
	\item $\Rank_a(Q,i)$: returns $|Q[{}:i]|_a$ for $1 \le i \le m$;
	\item $\Select_a(Q,i)$: returns $i$-th occurrence position of $a$ for $1 \le i \le \Rank_a(Q,m)$;
	\item $\Insert_a(Q,i)$: inserts $a$ between $Q[i-1]$ and $Q[i]$ for $1 \le i \le m+1$;
	\item $\Delete(Q,i)$: deletes $Q[i]$ from $Q$ for $1 \le i \le m$.
\end{enumerate}

\Cref{cor:Kim4} implies that we can compute $\LF_T(i)$ and its inverse $\LF_T^{-1}(j)$ by
\begin{align*}
	\LF_T(i) = \Select_{x}(\First_T, \Rank_{x}(\PBWT_T, i))  &\quad\text{where}\quad x = \PBWT_T[i],
\\	\LF_T^{-1}(j) = \Select_{y}(\PBWT_T, \Rank_{y}(\First_T, j)) &\quad\text{where}\quad y = \First_T[j].
\end{align*}

%% file: algorithm.tex
\section{Computing pBWT online}
This section introduces our algorithm computing $\PBWTF_T$ in an online manner by reading a p-string $T$ from right to left.
Let $T=cS$ for $c \in \Sigma \cup \Pi \setminus \{\$\}$ and $n=|S|\ge 1$.
We consider updating $\PBWT_S$ to $\PBWT_{T}$.
Hereafter, we assume that $\Sigma$ is known and $|\Sigma| \le |T|$ as in~\cite{policriti15fast}.
Among the rows of the rotation matrices of $S$ and $T$, the rows of $\llrrbracket{S}=\llrrbracket{S_n}$ and $\llrrbracket{T}=\llrrbracket{T_{n+1}}$ play important roles when updating.
Let $\Dol = \RA_{S}^{-1}[n]$ and $\Pos = \RA_{T}^{-1}[n+1]$.
We note that $\PBWT_S[\Dol] = \PBWT_T[\Pos] = \$$.

First, we observe $\RA_T$ is obtained from $\RA_{S}$ just by ``inserting'' $n+1$ at $k_T$. 
\begin{lemma}\label{cor:RAT}
For $1 \le i \le n+1$,
	\[
	\RA_T[i] = \begin{cases}
		\RA_S[i]	&\text{if $i < \Pos$,}
\\		n+1		&\text{if $i = \Pos$,}
\\		\RA_S[i-1]	&\text{if $i > \Pos$.}
	\end{cases}
	\]	
\end{lemma}
In the BWT, where $S$ and $T$ have no parameter characters, this implies that $\PBWT_{T}[i]=T_{\RA_T[i]}[n+1]=S_{\RA_S[i]}[n]=\PBWT_{S}[i]$ for $i < \Pos$ and $\PBWT_{T}[i+1]=T_{\RA_T[i+1]}[n+1]=S_{\RA_S[i]}[n]=\PBWT_{S}[i]$ for $i > \Pos$, except when $i=\Dol$.
Therefore, for computing $\PBWT_{T}$ from $\PBWT_S$, we only need to update $\PBWT_S[\Dol]=\$$ to $c$ and to find the position $\Pos=\RA_{T}^{-1}[n+1]$ where $\$$ should be inserted.
However in the pBWT, $\RA_{T}[i]=\RA_S[i]$ does not necessarily imply that the values $\PBWT_{T}[i]=\llrrbracket{T_{\RA_{T}[i]}}[n+1]$ and $\PBWT_{S}[i]=\llrrbracket{S_{\RA_S[i]}}[n]$ coincide, since it is not always true that $\llrrbracket{S} = \llrrbracket{T}[2:{}]$.
So we also need to update the values of the encoding.

\begin{algorithm2e}[t!]
	\caption{PBWT update algorithm}
	\label{alg:updateall}
	\SetVlineSkip{0.5mm}
	\Fn{$\UpdateAll(c, n, \PBWT, \First, \SFR, \SFL, \RM, \SFC, \LCP)$}{
		$k = \Select_{\$}(\PBWT, 1)$\label{ln:dol}\tcp*{$= k_S$}
		$\PBWT, \First, \SFR, \SFL, \RM = \UpdateLF(c, n, \PBWT, \First, \SFR, \SFL, \RM, k)$%
		\label{ln:updateLF}\tcp*{$=\PBWT^\circ_T,\First^\circ_T,\SFR_T,\SFL_T,\RM^\circ_T$}
		$\PBWT, \First, \SFC, k' = \InsertDol(n, \PBWT, \First, \SFC, k)$%
		\label{ln:insertion}\tcp*{$=\PBWT_T,\First_T,\SFC_T,k_T$}
		\ForEach{$a \in \Pi$}{%
			\textbf{if} $\RM[a] \ge k'$ \textbf{then}
				$\RM[a] = \RM[a] + 1$\tcp*{$=\RM_T[a]$}
		}
		$x = \UpdateLCP(\PBWT, \First, \LCP, k')$%
		\tcp*{$=\LCP_T[k_T]$}
		$\LCP[k'-1] = \UpdateLCP(\PBWT, \First, \LCP, k'-1)$%
		\tcp*{$=\LCP_T[k_T-1]$}
		$\Insert_x(\LCP,k')$\;
		\textbf{return} $n + 1, \PBWT, \First, \SFR, \SFL, \RM, \SFC, \LCP$\;
	}
\end{algorithm2e}

\Cref{alg:updateall} shows our update procedure, which maintains the array $\First$ and other auxiliary data structures in addition to $\PBWT$.
After getting the key position $\Dol$ as the unique occurrence position of $\$$ in $\PBWT_S$ at \Cref{ln:dol},
to update the values of $\PBWT$ and $\First$ from $\PBWT_S$ and $\First_S$ to $\PBWT_T$ and $\First_T$, respectively,
we compute intermediate arrays $\PBWT^\circ_T$ and $\First_T^\circ$ of length $n$, which satisfy
\begin{align*}
	\PBWT^{\circ}_{T}[i] &= \llrrbracket{T_{\RA_S[i]}}[n+1]
\quad \text{ and } \quad
\First^{\circ}_{T}[i] = \llrrbracket{T_{\RA_S[i]}}[1]
\end{align*}
for $1 \le i \le n$ using $\UpdateLF$ at \Cref{ln:updateLF}.
In other words, $\PBWT^{\circ}_{T}$ and $\First^{\circ}_{T}$ are extracted from the last and the first columns of the $n\times(n+1)$ matrix $(\llrrbracket{T_{\RA_S[i]}})_{i=1}^n$, respectively, which can conceptionally be obtained by deleting the $\Pos$-th row of the rotation sort matrix of $T$.
We then find the other key position $\Pos$ and inserts appropriate values into $\PBWT^{\circ}_{T}$ and $\First^{\circ}_{T}$ at $\Pos$ to turn them into $\PBWT_{T}$ and $\First_{T}$, respectively, by $\InsertDol$ at \Cref{ln:insertion}.
The rest of the algorithm is devoted to maintaining some of the helper arrays. 
Particularly, a dedicated function $\UpdateLCP$ is used to update the $\inft$-LCP array.
In the remainder of this section, we will explain those functions and involved auxiliary data structures in respective subsections.
\Cref{table:update} shows an example of our 2-step update.

\begin{table}[t]
	\caption{An example of our update step for $S=\mathtt{xayzzazyza\$}$ and $T=\mathtt{y}S$. 
	The updated and inserted values are highlighted. 
	In the arrays $\lrangle{T_{\RA_S[i]}}$, updated/inserted values appear only after $\$$.
	\Cref{cor:RAT,lem:equallcp} are immediate consequences of this observation.
	}
	\label{table:update}
		\centering
		\scalebox{\scl}{
			\begin{tabular}{|c|c|c|}
				\hline
				$\First_S[i]$ & $\lrangle{S_{\RA_S[i]}}$ & $\PBWT_S[i]$ \\
				\hline
				$\bdl$   & $\bdl\ift\tta\ift\ift1\tta252\tta$ & $\tta$ \\
				$\tta$ & $\tta\bdl\ift\tta\ift\ift1\tta252$ & $3$    \\
				$\tta$ & $\tta\ift\ift1\tta252\tta\bdl\ift$ & $3$    \\
				$\tta$ & $\tta\ift\ift2\tta\bdl\ift\tta661$ & $1$    \\
				$3$    & $\ift\tta\bdl\ift\tta\ift61\tta25$ & $3$    \\
				$1$    & $\ift\tta2\ift2\tta\bdl\ift\tta66$ & $1$    \\
				$3$    & $\ift\tta\ift\ift1\tta252\tta\bdl$ & $\bdl$   \\
				$1$    & $\ift1\tta2\ift2\tta\bdl\ift\tta6$ & $2$    \\
				$3$    & $\ift\ift\tta\bdl\ift\tta661\tta2$ & $2$    \\ \hdashline[2pt/1pt]
				$2$    & $\ift\ift1\tta252\tta\bdl\ift\tta$ & $\tta$ \\
				$2$    & $\ift\ift2\tta\bdl\ift\tta661\tta$ & $\tta$ \\
				\hline 
				\multicolumn{3}{c}{\phantom{A}}
			\end{tabular}
		}
	\quad
		\scalebox{\scl}{
			\begin{tabular}{|c|c|c|}
				\hline
				$\First^{\circ}_{T}[i]$ & $\lrangle{{T}_{\RA_{S}[i]}}$ & $\PBWT^{\circ}_{T}[i]$ \\
				\hline
				$\bdl$   & $\bdl\rd{\ift}\ift\tta\rd{3}\ift1\tta252\tta$ & $\tta$ \\
				$\tta$ & $\tta\bdl\rd{\ift}\ift\tta\rd{3}\ift1\tta252$ & $3$    \\
				$\tta$ & $\tta\ift\ift1\tta252\tta\bdl\rd{4}\ift$ & $3$    \\
				$\tta$ & $\tta\ift\ift2\tta\bdl\rd{4}\ift\tta\rd{3}61$ & $1$    \\
				$3$    & $\ift\tta\bdl\rd{\ift}\ift\tta\rd{3}61\tta25$ & $\rd{2}$    \\
				$1$    & $\ift\tta2\ift2\tta\bdl\rd{4}\ift\tta\rd{3}6$ & $1$    \\
				$3$    & $\ift\tta\ift\ift1\tta252\tta\bdl\rd{4}$ & $\rd{2}$   \\
				$1$    & $\ift1\tta2\ift2\tta\bdl\rd{4}\ift\tta\rd{3}$ & $2$    \\
				$\rd{2}$    & $\ift\ift\tta\bdl\rd{4}\ift\tta\rd{3}61\tta2$ & $2$    \\   \hdashline[2pt/1pt]
				$2$    & $\ift\ift1\tta252\tta\bdl\rd{4}\ift\tta$ & $\tta$ \\
				$2$    & $\ift\ift2\tta\bdl\rd{4}\ift\tta\rd{3}61\tta$ & $\tta$ \\
				\hline 
				\multicolumn{3}{c}{\phantom{A}}
			\end{tabular}
		}
	\quad
		\scalebox{\scl}{
			\begin{tabular}{|c|c|c|}
				\hline
				$\First_{T}[i]$ & $\lrangle{{T}_{\RA_{T}[i]}}$ & $\PBWT_{T}[i]$ \\
				\hline
				$\bdl$   & $\bdl\ift\ift\tta3\ift1\tta252\tta$ & $\tta$ \\
				$\tta$ & $\tta\bdl\ift\ift\tta3\ift1\tta252$ & $3$    \\
				$\tta$ & $\tta\ift\ift1\tta252\tta\bdl4\ift$ & $3$    \\
				$\tta$ & $\tta\ift\ift2\tta\bdl4\ift\tta361$ & $1$    \\
				$3$    & $\ift\tta\bdl\ift\ift\tta361\tta25$ & $2$    \\
				$1$    & $\ift\tta2\ift2\tta\bdl4\ift\tta36$ & $1$    \\
				$3$    & $\ift\tta\ift\ift1\tta252\tta\bdl4$ & $2$   \\
				$1$    & $\ift1\tta2\ift2\tta\bdl4\ift\tta3$ & $2$    \\
				$2$    & $\ift\ift\tta\bdl4\ift\tta361\tta2$ & $2$    \\ \hdashline[2pt/1pt]
				$\rd{2}$    & $\rd{\ift\ift\tta3\ift1\tta252\tta\bdl}$ & $\rd{\bdl}$    \\ \hdashline[2pt/1pt]
				$2$    & $\ift\ift1\tta252\tta\bdl4\ift\tta$ & $\tta$ \\ 
				$2$    & $\ift\ift2\tta\bdl4\ift\tta361\tta$ & $\tta$ \\
				\hline
			\end{tabular}
		}
\end{table}

\subsection{Step 1: $\UpdateLF$ computes $\PBWT^{\circ}_{T}[i]$ and $\First^{\circ}_{T}[i]$}

When $c \in \Sigma$, computing $\PBWT_T^\circ$ and $\First_T^\circ$ from $\PBWT_S$ and $\First_S$, respectively, is easy.
\begin{restatable}{lemma}{lemstatic}
	\label{lem:static}
	If $c \in \Sigma$, then for any $i \in \{1,\dots,n\}$,
	$\First^{\circ}_{T}[i] = \First_S[i]$ and $ \PBWT^{\circ}_{T}[i] = \PBWT_S[i]$ except for $\PBWT^{\circ}_{T}[\Dol] = c$.
\end{restatable}
Concerning the case $c \in \Pi$, first let us express the values of $\llrrbracket{T}$ using $\llrrbracket{S}$.
\begin{restatable}{lemma}{lemparameterized}
\label{lem:parameterized}
	Suppose $c \in \Pi$.
	\[
		\llrrbracket{T}[1] = \begin{cases}
			|\rest{\Pi}{S}|+1	& \text{if $\CL_S(c) = 0$,}
	\\		|\rest{\Pi}{S[1:\CL_S(c)]}|	 & \text{otherwise.}
		\end{cases}
	\]
	For $1 \le p \le n$, if $S[p] \in \Sigma$ or $p \ne \CR_S(S[p])$, then $\llrrbracket{T}[p+1] = \llrrbracket{S}[p]$.
	If $S[p] = a \in \Pi$ and $p = \CR_S(a)$, then
	\begin{align*}
		\llrrbracket{T}[p+1] = 
		\begin{cases}
			|\rest{\Pi}{S[p+1:n]}|+1 \quad \text{if } a = c, \\
			\llrrbracket{S}[p]+1 \quad \text{if }  \CL_S(c) = 0 \text{ or } \\
			\text{\phantom{$\llrrbracket{S}[p]+1$\quad if }} \CL_S(a) < \CL_S(c) \le \CR_S(c) < \CR_S(a), \\
			\llrrbracket{S}[p] \text{\phantom{${}+1$\quad}} \text{otherwise}.
		\end{cases}
	\end{align*}
\end{restatable}
\begin{algorithm2e}[t!]
	\caption{Computing $\PBWT^{\circ}_{T}$ and $\First^{\circ}_{T}$}
	\label{alg:update1}
	\SetVlineSkip{0.5mm}
	\Fn{$\UpdateLF(c, n, \PBWT, \First, \SFR, \SFL, \RM, k)$}{
		\lIf{$c \in \Sigma$}{%
			$\PBWT[k] = c$}
		\Else{
			\ForEach{$a \in \Pi$ with $\SFL[a] \ne 0$}{ \label{line1:for1}
				\tcp{Computing $\PBWT^{\circ}_{T}[\RM[a]]=\First^{\circ}_{T}[\LF_S(\RM[a])]$}
					$i = \RM[a]$\;
					$j = \Select_{\PBWT[i]}(\First, \Rank_{\PBWT[i]}(\PBWT, i))$\tcp*{$j=\LF_S(i)$}
					\uIf{$a = c$}{
						$\Cnt=0$\;
						\ForEach{$b \in \Pi$ with $\SFL[b] \ne 0$}{\label{line1:for2}
							\lIf{$\SFL[a] \ge \SFR[b]$}{%
								$\Cnt = \Cnt+1$%
							}
						}
					}
					\Else{
						$\Cnt = \PBWT[i]$\;
						\If{$\SFL[c] = 0$  or  $\SFL[a] > \SFL[c] \ge \SFR[c] > \SFR[a]$}{
							$\Cnt = \Cnt+1$\;
						}
					}
					$\PBWT[i] = \Cnt$;
					$\First[j] = \Cnt$\;
			}
			\tcp{Computing $\PBWT^{\circ}_{T}[\Dol]=\llrrbracket{T}[1]$}
			$\Cnt=1$\;
			\uIf{$\SFL[c] = 0$}{
				\lForEach{$a \in \Pi$ with $\SFL[a] \ne 0$}{%
						$\Cnt = \Cnt+1$}
				$\SFR[c] = n+1$;
				$\SFL[c] = n+1$;
				$\RM[c] = k$\;\label{ln:RM}
			}\Else{
				\lForEach{$a \in \Pi$ with $\SFL[a] > \SFL[c]$}{%
						$\Cnt = \Cnt+1$}
				$\SFL[c] = n+1$\;
			}
			$\PBWT[k] = \Cnt$\;
		}
	\textbf{return} $\PBWT, \First, \SFR, \SFL, \RM$\;
	}
\end{algorithm2e}

Based on \Cref{lem:static,lem:parameterized}, \Cref{alg:update1} computes $\First^{\circ}_{T}[i]$ and $\PBWT^{\circ}_{T}[i]$ from $\First_{S}[i]$ and $\PBWT_{S}[i]$, as well as other auxiliary data structures.
Note that, since the intermediate matrix $(T_{\RA_S[i]})_{i=1}^n$ misses a row corresponding to $\llrrbracket{T}$, the value $\llrrbracket{T}[1]$ does not matter for $\First^\circ_T$, whereas it appears as $\PBWT_T^\circ[\Dol]=\PBWT^\circ_T[\RA_S^{-1}[n]]$.
When $c \in \Pi$, \Cref{lem:parameterized} implies that, other than $\PBWT^\circ_T[\Dol]=\llrrbracket{T}[1]$, we only need to update the values at the positions in $\PBWT$ and $\First$ corresponding to the rightmost occurrence position $p = \CR_S(a)$ of each parameter character $a \in \Pi$ in $S$.
By rotating $S$ to the right by $n-p$, that occurrence comes to the right end and appears in the pBWT.
That is, the array $\PBWT$ needs to be updated only at $i$ such that $\RA_S[i] = n - \CR_S(a)$.
The algorithm maintains such position $i$ as $\RM_S[a]$ for each $a \in \rest{\Pi}{S}$, i.e. $\RM_S[a] = \RA^{-1}_S[n - \CR_S(a)]$.
Similarly, we only need to update $\First$ at $\LF_S(\RM_S[a])$, where $\First_T^\circ[\LF_S(\RM_S[a])]=\PBWT_T^\circ[\RM_S[a]]$.
In our algorithm, as alternatives of $\CL_S$ and $\CR_S$, we maintain two arrays $\SFL$ and $\SFR$ that store the leftmost and rightmost occurrence positions of parameter characters counting \emph{from the right end}, respectively, i.e., $\SFL_S[a] = \CR_{\overline{S}}(a)$ and $\SFR_S[a] = \CL_{\overline{S}}(a)$ for each $a \in \Pi$, where $\overline{S}$ is the reverse of $S$.

\Cref{alg:update1} also updates $\RM$ to $\RM_T^\circ$, which indicates the row of $\PBWT_T^\circ$ corresponding to the rightmost occurrence of each parameter character in $T$.
That is, $\RM_T^\circ[a] = i$ iff $\RA_S[i] = \SFR_T[a]$, as long as $a$ occurs in $T$.
When $c \in \Pi$ and it appears in the text for the first time, we have $\RM_T^\circ[c] = \Dol$ (\Cref{ln:RM}).
Other than that, $\RM_T^\circ[a]=\RM_S[a]$ for every $a \in \Pi$.

\begin{restatable}{lemma}{lemupdateone}
\label{lem:update1}
	\Cref{alg:update1} computes $\PBWT^{\circ}_{T}[i]$, $\First^{\circ}_{T}[i]$, $\SFR_{T}$, $\SFL_{T}$, and\/ $\RM^\circ_T$ in \linebreak $O(|\Pi| \frac{\log n}{\log \log n})$ amortized time.
\end{restatable}

\subsection{Step 2: $\InsertDol$ computes $\PBWT_T$ and $\First_T$}

To transform $\First^\circ_T$ and $\PBWT^\circ_T$ into $\First_T$ and $\PBWT_T$, we insert the values $\llrrbracket{T}[1]$ and $\llrrbracket{T}[n+1]$ at the position $\Pos$, respectively.
We know those values as $\llrrbracket{T}[1]=\PBWT_T^\circ[\Dol]$ and $\llrrbracket{T}[n+1]=\$$.
Therefore, it is enough to discuss how to find the position $\Pos$.

In the case $c \in \Sigma$, the position $\Pos$ can be calculated similarly to the case of BWT for static strings thanks to \Cref{cor:Kim4}.
Define $\Sigma_{< b} = |\{\, a \in \Sigma \mid a < b \,\}|$.
\begin{restatable}{lemma}{lempostatic}
\label{lem:postatic}
If $c \in \Sigma$,
\(
	\Pos = |T|_{\Sigma_{< c}} + |\{\, i \mid \PBWT^{\circ}_{T}[i] = c,\, 1 \le i \le \Dol \,\}| 
\).
\end{restatable}
In the case $c \in \Pi$, we will use \Cref{lem:cross} for finding $\Pos$ in \Cref{lem:posparameter} below.
We first observe that one can use $\LCP_S$ to calculate $\Lcp(\lrangle{T_p}, \lrangle{T_q})$ for most cases.
\begin{restatable}{lemma}{lemequallcp}
\label{lem:equallcp}
	For  $1 \le p < q \le n$,
		$\Lcp(\lrangle{T_p}, \lrangle{T_q}) = \Lcp(\lrangle{S_p},\linebreak[1] \lrangle{S_q})$.
\end{restatable}
\begin{lemma}\label{lem:posparameter}
	Suppose $c \in \Pi$.
	Let $\ell_i = \Lcp(\lrangle{S_{\RA_S[i]}},\lrangle{S_{\RA_S[\Dol]}})$ for $1 \le i \le n$.
	Then,
	\begin{align}
\notag	\Pos ={} & 1 + |T|_\Sigma
\\	& + |\{\, i \mid 1 \le \PBWT^{\circ}_{T}[i] \le \PBWT^{\circ}_{T}[\Dol],\  1 \le i < \Dol \,\}| \label{eq:pospara3}
\\	& + |\{\, i \mid \ell_i < \PBWT^{\circ}_{T}[\Dol] < \PBWT^{\circ}_{T}[i],\ 1 \le i < \Dol \,\}| \label{eq:pospara4} 
\\	& + |\{\, i \mid 1 \le \PBWT^{\circ}_{T}[i] \le \min\{ \PBWT^{\circ}_{T}[\Dol]-1, \ell_i \},\ \Dol < i \le n \,\}| \label{eq:pospara5}
	\,.
	\end{align}
\end{lemma}
\begin{proof}
By definition,
\[
	\Pos = 1 + |T|_\Sigma + |\{\, j \mid \First_{T}[j] \in \mbb{N}_+,\, 1 \le j < \Pos \,\}| 
\,,\]
of which we focus on the last term.
Let $h=\LF_T^{-1}(\Pos)$ and $m_i = \Lcp(\lrangle{T_{\RA_{T}[i]}},\linebreak[1]\lrangle{T_{\RA_{T}[h]}})$ for $1 \le i \le n+1$.
By \Cref{lem:cross}, $\First_{T}[j] \in \mbb{N}_+$ and $1 \le j < \Pos$ iff for $i = \LF_T^{-1}(j)$, either
\begin{enumerate}
	\item $1 \le i < h$ and $1 \le \PBWT_{T}[i] \le \PBWT_{T}[h]$,
	\item $1 \le i < h$ and $m_i < \PBWT_{T}[h] < \PBWT_{T}[i]$, or
	\item $h < i \le n+1$ and $1 \le \PBWT_{T}[i] \le \min\{ \PBWT_{T}[h]-1,\, m_i\} $.
\end{enumerate}
Those three cases are mutually exclusive.
Let $m^\circ_i = \Lcp(\lrangle{T_{\RA_S[i]}},\lrangle{T_{\RA_S[\Dol]}})$.
Counting each of the above cases is equivalent to counting $i$ such that
\begin{enumerate}
	\item $1 \le i < \Dol$ and $1 \le \PBWT^{\circ}_{T}[i] \le \PBWT^{\circ}_{T}[\Dol]$,
	\item $1 \le i < \Dol$ and $m^\circ_i < \PBWT^{\circ}_{T}[\Dol] < \PBWT^{\circ}_{T}[i]$, or
	\item $\Dol < i \le n$ and $1 \le \PBWT^{\circ}_{T}[i] \le \min\{ \PBWT^{\circ}_{T}[\Dol]-1,\, m^\circ_i\} $.
\end{enumerate}
This is because the matrix $(\llrrbracket{T_{\RA_S[i]}})_{i=1}^n$ can conceptionally be obtained by removing the $\Pos$-th row of the matrix of $(\llrrbracket{T_{\RA_T[i]}})_{i=1}^{n+1}$, where the row $\Dol$ of $(\llrrbracket{T_{\RA_S[i]}})_{i=1}^n$ corresponds to the row $h$ of $(\llrrbracket{T_{\RA_T[i]}})_{i=1}^{n+1}$ in particular ($\RA_S[\Dol]=\RA_T[h]=n$),
and $i=\Pos=\RA^{-1}_{T}[n+1]$ is not counted due to ${T}[n+1] = \$ \in \Sigma$.

\Cref{lem:equallcp} implies $m_i^\circ=\ell_i$, which completes the proof.
\qed\end{proof}

\begin{algorithm2e}[t]
	\caption{Inserting $\llrrbracket{T}[1]$ to $\First$ and $\llrrbracket{T}[n+1]$ to $\PBWT$}
	\label{alg:update2}
	\SetVlineSkip{0.5mm}
	\Fn{$\InsertDol(n, \PBWT, \First, \SFC, k)$}{
		$x = \PBWT[k]$\tcp*{$=\llrrbracket{T}[1]$}
		\uIf{$x \in \Sigma$}{
			$k' = \Select_x(\SFC,1) - |\Sigma_{< x}| - 1 + \Rank_x(\PBWT, k)$\tcp*{$\Pos = |S|_{\Sigma_{< c}} + |\{\, i \mid \PBWT^{\circ}_{T}[i] = c,\, 1 \le i \le \Dol \,\}|$}
			$\Insert_x(\SFC,\Select_x(\SFC,1))$\;
		}
		\Else{
			$k' = 1+|\SFC|-|\Sigma|$\tcp*{$= 1+|S|_\Sigma$}
			\lFor*{$y = 1$ \textbf{to} $x$}{%
				$k' = k' + \Rank_y(\PBWT, k-1)$\tcp*[r]{Term (1)}\label{ln:pos1}}
			$j = 0$\;
			\For{$y = 0$ \textbf{to} $x-1$\label{ln:range1}}{%
				\If{$\Rank_y(\LCP, k-1) \ne 0$}{
					$j = \max\{j,\, \Select_y(\LCP,\Rank_y(\LCP, k-1))\}$%
					\tcp*{$j = \max(\{j\} \cup \{\,i \mid \LCP[i]=y \text{ and } 1 \le i < \Dol \,\})$}
				}
			}
			\lFor*{$y = x+1$ \textbf{to} $|\Pi|$}{%
				$k' = k' + \Rank_y(\PBWT, j)$\tcp*[r]{Term (2)}}\label{ln:pos2}
			$j = n$\tcp*{$j_0=n$}
			\For(\tcp*[f]{Term (3)}){$y = 1$ \textbf{to} $x-1$\label{ln:pos3}}{
				\If{$\Rank_{y-1}(\LCP, k-1) < \Rank_{y-1}(\LCP, n)$}{
					$j = \min\{j,\, \Select_{y-1}(\LCP,\Rank_{y-1}(\LCP, k-1)+1) \}$\label{ln:jy}%
					\tcp*{$j_y = \min(\{j_{y-1}\} \cup \{\,i \mid \LCP[i]=y-1 \text{ and } \Dol \le i \le n \,\})$}
					$k' = k' + \Rank_y(\PBWT, j-1) - \Rank_y(\PBWT, k)$\;
				}
			}
		}
		$\Insert_{\$}(\PBWT,k')$;
		$\Insert_{x}(\First,k')$\;
		\textbf{return} $\PBWT, \First, \SFC, k'$\;
	}
\end{algorithm2e}

Based on \Cref{lem:postatic,lem:posparameter}, \Cref{alg:update2} finds the key position $\Pos$.

For handling the case $c \in \Sigma$, we maintain a dynamic array $\SFC$ by which one can obtain the value $|T|_{\Sigma_{< c}}=|S|_{\Sigma_{< c}}$ quickly.
The array $\SFC_S$ can be seen as a string of the form
\(
	\SFC_S = a_1^{|S|_{a_1}+1} \dots a_\sigma^{|S|_{a_\sigma}+1} 
\),
where $a_1,\dots,a_\sigma$ enumerate the static characters of $\Sigma$ in the lexicographic order ($\sigma=|\Sigma|$) and $a^s$ denotes the sequence of $a$ of length $s$.
Then, $|T|_{\Sigma_{< c}} = \Select_c(\SFC_S,1) - |\Sigma_{< c}| - 1$.
The other term $|\{\, i \mid \PBWT^{\circ}_{T}[i] = c,\, 1 \le i \le \Dol \,\}|$ in \Cref{lem:postatic} is calculated as $\Rank_c(\PBWT_T^\circ, k_S)$.
We remark $\SFC_S$ has $a^{|S|_{a}+1}$ rather than $a^{|S|_{a}}$ so that $\Select_c(\SFC,1)$ is always defined.

Suppose $c \in \Pi$.
The term $|T|_\Sigma$ of the equation of \Cref{lem:posparameter} is calculated as $|T|_\Sigma = |S|_\Sigma = |\SFC|-|\Sigma|$. 
Let $x = \PBWT^{\circ}_{T}[\Dol]$.
Term~(\ref{eq:pospara3}) is obtained at \Cref{ln:pos1} by 
\[
\text{(\ref{eq:pospara3})}
= \textstyle\sum_{y=1}^x \Rank_y(\PBWT_T^\circ,\Dol-1)
\,.\]
Concerning Term~(\ref{eq:pospara4}), we first find the range of $i < \Dol$ satisfying $\ell_i < x$.
By \Cref{lem:lcpmin}, $\ell_i = \min_{i \le j < \Dol} \LCP_S[j]$.
Thus, for any $i < \Dol$, $\ell_i < x$ iff $i \le j_* = \max\{\,j \mid \LCP_S[j] < x,\ j < \Dol \,\}$.
The \textbf{for} loop of \Cref{ln:range1} computes such $j_*$.
Then, (\ref{eq:pospara4}) is computed  at \Cref{ln:pos2} as 
\begin{align*}
\text{(\ref{eq:pospara4})}
 =  |\{\, i \mid x < \PBWT^{\circ}_{T}[i],\ 1 \le i \le j_* \,\}|
= \textstyle\sum_{y=x+1}^{|\Pi|} \Rank_y(\PBWT^\circ_T, j_*)
\,.\end{align*}
We compute Term~\text{(\ref{eq:pospara5})} by summing up the numbers of positions $i > \Dol$ such that $\PBWT_T^\circ[i] = y \le \ell_i$ for all $y =1,\dots, x-1$ in the \textbf{for} loop of \Cref{ln:pos3}.
To this end, we find the range of $i > \Dol$ such that $\ell_i \ge y$.
By \Cref{lem:lcpmin}, $\ell_i = \min_{\Dol \le j < i} \LCP_S[j]$.
Thus, for every $i > \Dol$, $\ell_i \ge y$ iff $i < j_y = \min \{\,j \mid \LCP_S[j] < y,\ \Dol \le j \le n \,\}$.
Note that $j_{y} = \min(\{j_{y-1}\} \cup \{\,j \mid \LCP_S[j] = y-1,\ \Dol \le j \le n \,\})$ for any $y \ge 1$ assuming $j_0=n$.
Line~\ref{ln:jy} computes $j_y$ as the first occurrence of $y-1$ after those in $\LCP[1:\Dol-1]$.
Then, (\ref{eq:pospara5}) is calculated by
\begin{align*}
\text{(\ref{eq:pospara5})}
	&= |\{\, i \mid 1 \le \PBWT^{\circ}_{T}[i] \le x-1, \ \Dol < i < j_y \,\}|
\\	&= \textstyle\sum_{y=1}^{x-1} \big(\Rank_y(\PBWT^\circ_T, j_y-1) - \Rank_y(\PBWT^\circ_T, \Dol) \big)
\,.\end{align*}

\begin{restatable}{lemma}{lemupdatetwo}
	\label{lem:update2}
	\Cref{alg:update2} computes $\PBWT_{T}$, $\First_{T}$, $\SFC_T$, and $k_T$ in $O(|\Pi| \frac{\log n}{\log \log n})$ amortized time.
\end{restatable}

\subsection{Step 3: Updating $\LCP$ by $\UpdateLCP$}
What remains to do is updating the arrays $\RM$ and $\LCP$.
On the one hand, updating $\RM$ from $\RM_T^\circ$ to $\RM_T$ is easy.
$\RM_T^\circ[a]$ should be incremented by one just if $\RM_T^\circ[a] \ge k_T$.
Otherwise, $\RM_T[a]=\RM_T^\circ[a]$.
On the other hand, \Cref{lem:equallcp} implies $\LCP_{T}$ is almost identical to $\LCP_{S}$.
\begin{corollary}\label{cor:lcparray}
	$\LCP_{T}[i] = \LCP_{S}[i]$ if $i < \Pos - 1$, and 
	$\LCP_{T}[i] = \LCP_{S}[i-1]$ if $i > \Pos$.
\end{corollary}
By \Cref{cor:lcparray}, we only need to compute $\LCP_{T}[\Pos-1]$ and $\LCP_{T}[\Pos]$, to which \Cref{lem:equallcp} cannot directly be applied.
The following lemma allows us to reduce the calculation of $\LCP_{T}[k] = \Lcp(\lrangle{T_{\RA_T[k]}},\lrangle{T_{\RA_T[k+1]}})$
to that of $\Lcp(\lrangle{T_{\RA_T[\LF_T^{-1}(k)]}},\lrangle{T_{\RA_T[\LF_T^{-1}(k+1)]}})$, to which \Cref{lem:equallcp} may be applied.
\begin{lemma}\label{lem:updatelcparray}
	Let $1 \le i,j \le n+1$,
	$p = \RA_{T}[i]$, $q = \RA_{T}[j]$,
	$\ell = \Lcp(\lrangle{T_{p}},\lrangle{T_{q}})$,
	$i' \!= \LF_T^{-1}(i)$, $j' \!= \LF_T^{-1}(j)$,
 	$p' \!= \RA_{T}[i']$, $q' \!= \RA_{T}[j']$,
 	 and $\ell' \!= \Lcp(\lrangle{T_{p'}},\lrangle{T_{q'}})$.
	\begin{enumerate}
		\item If $\First_T[i] = \First_T[j]  \in \Sigma$, then $\ell = \ell'$.
		\item If $\First_T[i] \ne \First_T[j]$ and either $\First_T[i] \in \Sigma$ or $\First_T[j]  \in \Sigma$, then $\ell = 0$.
		\item If $\First_T[i], \First_T[j] \in \mbb{N}_+$, then
			\begin{align*}
				\ell = 
				\begin{cases}
					\ell' + 1 
					& \text{if } \ell' < \min\{\First_T[i],\First_T[j]\} , \\
					\ell'        
					& \text{if } \ell' \ge \First_T[i] = \First_T[j],\\
					\min\{\First_T[i], \First_T[j]\}  		
					& \text{otherwise}.
				\end{cases}
			\end{align*}		
	\end{enumerate}
\end{lemma}
\begin{proof}
	Let $T_p = aU$, $T_q = bV$, $T_{p'} = Ua$ and $T_{q'} = Vb$.
	
1. If $a = b \in \Sigma$, then $\ell = \ell' = \Lcp(\lrangle{U},\lrangle{V})$.

2. In the case $a \ne b$ and $\{a, b\} \cap \Sigma \neq \emptyset$, clearly $\lrangle{T_{p}}[1] \ne \lrangle{T_{q}}[1]$. Thus $\ell = 0$.
		
3. In the case $a,b \in \Pi$, let $W$ be the longest common prefix of $\lrangle{T_{p'}}$ and $\lrangle{T_{q'}}$,
		 $u$ and $v$ be the first occurrence positions of $a$ in $T_{p'}$ and $b$ in $T_{q'}$, respectively, and $w$ be the $\ell'$-th occurrence position of $\inft$ in $W$.

		Suppose $\ell' < \min\{\First_T[i],\First_T[j]\}=\min\{\PBWT_T[i'],\PBWT_T[j']\}$. That is, $|W|_\inft \linebreak[2]<\linebreak[4] \min\{|\lrangle{T_{p'}}[:u]|_\inft,\linebreak[2]|\lrangle{T_{q'}}[:v]|_\inft\}$.
		This means $|W| < \min\{u,v\}$ and thus $\inft W$ is the longest common prefix of $\lrangle{T_{p}}$ and $\lrangle{T_{q}}$.
		Thus, we have $\ell = |\inft W|_\inft = \ell' + 1$.

		Suppose $\ell' \ge \First_T[i] = \First_T[j]$, i.e., $|W|_\inft \ge |\lrangle{T_{p'}}[:u]|_\inft = |\lrangle{T_{q'}}[:v]|_\inft$.
		Then $|W[:u]|_\inft = |W[:v]|_\inft$ and $W[u]=W[v]=\inft$ implies $u=v$. 
		Let $Z$ be the longest common prefix of $\lrangle{T_{p}}$ and $\lrangle{T_{q}}$.
		Then, $W$ and $Z$ can be written as $W = X \ift Y$ and $Z = \ift X u Y$, where $|X|=u-1$.
		Therefore, $\ell = \ell'$.
		
		Otherwise, $\First_T[i] \ne \First_T[j]$ and $\ell' \ge \min\{\First_T[i], \First_T[j]\}$.
		Assume $\First_T[i] < \First_T[j]$ (the case $\First_T[j] < \First_T[i]$ is symmetric).
		Then $u \le |W|$.
		Moreover, we have $u < v$, since otherwise, $\lrangle{T_{q'}}[:v]$ had to be a prefix of $\lrangle{T_{p'}}[:u]$, which is impossible by $\First_T[i] < \First_T[j]$.
		Let $\lrangle{T_{p}[:|W|+1]} = \ift X u Y$, where $|X|=u-1$.
		Then we have $\lrangle{T_{q}[:|W|+1]} = \ift X \ift Y'$ for some $Y' \in (\Sigma \cup \mbb{N}_\inft)^*$.
		Thus $\ell = |\ift X|_\inft = \First_T[i]$.
		\qed
\end{proof}
One can compute $\ell'=\Lcp(\lrangle{T_{p'}},\lrangle{T_{q'}})$ in Lemma~\ref{lem:updatelcparray} for $1 \le p'<q' \le n$ using Lemmas~\ref{lem:equallcp} and~\ref{lem:lcpmin} as
\begin{multline*}
\Lcp(\lrangle{T_{p'}},\lrangle{T_{q'}}) =
	\Lcp(\lrangle{S_{p'}},\lrangle{S_{q'}})
	= \min\{\,\LCP_S[h] \mid i' \le h < j'\,\}
	\\ = \min(\{0\}\cup\{\,y \mid \Rank_y(\LCP_S, i'-1) \ne \Rank_y(\LCP_S, j'-1) \,\})\,.
\end{multline*}
Finally, when $q' = n+1$, we have $\First_T[j] =\$ \neq \First_T[i]$, and thus $\Lcp(\lrangle{T_{p}},\lrangle{T_{q}}) = 0$.
\Cref{alg:updatelcp} computes $\LCP_{T}[i]$ using $\First_T$, $\PBWT_T$, and $\LCP_{S}$.

\begin{algorithm2e}[t!]
	\caption{Updating $\LCP[i]$}
	\label{alg:updatelcp}
	\SetVlineSkip{0.5mm}
	\Fn{$\UpdateLCP(\PBWT, \First, \LCP, i)$}{
		$j = i + 1$;
		$x = 0$\;
		\uIf{$\First[i] = \First[j]$ or $\First[i],\First[j] \in \mbb{N}_+$}{%
			$i' = \Select_{\First[i]}(\PBWT, \Rank_{\First[i]}(\First,i))$\tcp*{$i'=\LF_{T}^{-1}(i)$\phantom{${}+1$}}
			$j' = \Select_{\First[j]}(\PBWT, \Rank_{\First[j]}(\First,j))$\tcp*{$j'=\LF_{T}^{-1}(i+1)$}
			\For{$y = |\Pi| $ \textbf{downto} $0$}{
				\lIf{$\Rank_y(\LCP, i'-1) \ne \Rank_y(\LCP, j'-1)$}{%
					$x = y$}
			}
			\tcp{$x = \Lcp(\lrangle{S_{\RA_S[i']}},\lrangle{S_{\RA_S[j']}})$}
			\uIf{$\First[i],\First[j] \in \mbb{N}_+$}{%
				\lIf{$x < \min\{\First[i], \First[j]\}$}{%
					$x = x + 1$%
				}
				\lElseIf{$\First[i] \ne \First[j]$}{%
					$x = \min\{\First[i], \First[j]\}$%
				}
			}
		}
		\textbf{return} $x$\;
	}
\end{algorithm2e}

\begin{restatable}{lemma}{lemupdatelcp}
	\label{lem:updatelcp}
	\Cref{alg:updatelcp} computes $\LCP_{T}[i]$ in $O(|\Pi| \frac{\log n }{ \log \log n})$ amortized time.
\end{restatable}

By Lemmas~\ref{lem:update1}, \ref{lem:update2}, and \ref{lem:updatelcp}, we have the following theorem.
\begin{theorem}
	Given $c \in \Sigma \cup \Pi$, $n=|S|$, $\PBWT=\PBWT_S$, $\First=\First_S$, $\SFR=\SFR_S$, $\SFL=\SFL_S$, $\RM=\RM_S$, $\SFC=\SFC_S$, and $\LCP=\LCP_S$ for some $S \in (\Sigma \cup \Pi)^*$,
	\Cref{alg:updateall} computes $|T|$, $\PBWT_T$, $\First_T$, $\SFR_T$, $\SFL_T$, $\RM_T$, $\SFC_T$, and $\LCP_T$ for $T=cS$ in $O(|\Pi| \frac{\log n}{\log \log n})$ amortized time per input character.
\end{theorem}
\begin{corollary}
	For a p-string $T$ of length $n$, $\PBWTF_{T}$ can be computed in an online manner by reading $T$ from right to left in $O(n|\Pi| \frac{\log n}{\log \log n})$ time.
\end{corollary}

%% file: appendix.tex
\section{Proofs}
\begin{proposition}\label{lem:llrrbracketEquiv}
For any p-strings $S$ and $T$,  $S \approx T$ if and only if $\llrrbracket{S}=\llrrbracket{T}$.
\end{proposition}
\begin{proof}
	For simplicity, assume that $S$ and $T$ contain no static character.
	Suppose $S \approx T$.
	Since $S[i]=S[j]$ iff $T[i]=T[j]$ for any indices $i,j$, we have
	\[
		\llrrbracket{S}[i]=|\rest{\Pi}{S_{n-i}[1:\CL_{S_{n-i}}(S[i])]}|
			=|\rest{\Pi}{T_{n-i}[1:\CL_{T_{n-i}}(T[i])]}|=\llrrbracket{T}[i]\,.
	\]
	for all $i$.

	Suppose $S \not\approx T$.
	Let $i$ be the leftmost position such that $\lrangle{S}[i] \neq \lrangle{T}[i]$.
	We may assume without loss of generality that $\lrangle{S}[i] < \lrangle{T}[i]$.
	Let $j = i-\lrangle{S}[i]$.
	Then,
	\[
	\llrrbracket{S}[j] = |\rest{\Pi}{S[j+1:i]}| = |\rest{\Pi}{S[j:i-1]}| = |\rest{\Pi}{T[j:i-1]}|
	\,,\]	
	since $S[j:i-1] \approx T[j:i-1]$.
	The fact $S[j] \notin \rest{\Pi}{S[j+1:i-1]}$ implies $T[j] \notin \rest{\Pi}{T[j+1:i-1]}$.
	Moreover, $\lrangle{S}[i] < \lrangle{T}[i]$ implies $T[i] \notin \rest{\Pi}{T[j:i-1]}$.
	Hence, 
	\[
	\llrrbracket{T}[j] \ge |(\rest{\Pi}{T[j:i-1]}) \cup \{T[i]\}| > |\rest{\Pi}{T[j:i-1]}| = \llrrbracket{S}[j]\,. 
	\qedhere\]
\end{proof}

\Cref{cor:RAT} is a corollary to the following lemma.
\begin{lemma}
\label{lem:RAT}
	For any $i$ and $j$ such that $1 \le i < j \le n$,
	$\RA^{-1}_S[i] < \RA^{-1}_S[j]$ iff $\RA^{-1}_{T}[i] < \RA^{-1}_{T}[j]$.
\end{lemma}
\begin{proof}
	Let $S_i = U\$V$ and $S_j = XY\$Z$, where $|U\$|=|X|=i<j \le n$.
	We have $T_i=U\$ c V$ and $T_j = XY\$ c Z$.
	Since $\$$ does not occur in $X$, $\lrangle{U \$}\neq\lrangle{X}$.
	Thus,
	\[
	\RA^{-1}_S[i] < \RA^{-1}_S[j] \iff \lrangle{U \$} < \lrangle{X} \iff \RA^{-1}_T[i] < \RA^{-1}_T[j] 
	\,. \qedhere\]
\end{proof}

\lemstatic*
\begin{proof}
	If $c \in \Sigma$, then $\llrrbracket{S} = \llrrbracket{T}[2:{}]$ by definition.
	So, for any $i \in \{1,\dots,n\}$,
	\[
	\First^{\circ}_{T}[i] = \llrrbracket{T_{\RA_S[i]}}[1] = \llrrbracket{S_{\RA_S[i]}}[1] = \First_S[i]
	\]
	and
	\[
	\PBWT^{\circ}_{T}[i] = \llrrbracket{T_{\RA_S[i]}}[n+1] = \begin{cases}
		\llrrbracket{S_{\RA_S[i]}}[n] = \PBWT_S[i]	& \text{if $\RA_S[i] \neq n$,}
		\\		c	& \text{if $\RA_S[i]=n$.}
	\end{cases}
	\]
	\qed
\end{proof}

\lemparameterized*
\begin{proof}
	The claim on value of $\llrrbracket{T}[1]$ is clear by definition.
	For $p \ge 1$, if $T[p+1]=S[p] \in \Sigma$, then $\llrrbracket{T}[p+1] = \llrrbracket{S}[p]$.
	
	Let us consider the case $S[p] = a \in \Pi$.
	If $p \ne \CR_S(a)$, then $a$ occurs somewhere after $p$ in $S$.
	Let $q > p$ be the first occurrence position of $a$ after $p$ in $S$.
	By definition, $\llrrbracket{S}[p]=|\rest{\Pi}{S[p+1:q]}| = |\rest{\Pi}{T[p+2:q+1]}|=\llrrbracket{T}[p+1]$.
	
	Suppose $p = \CR_S(a)$ for some $a \in \Pi$.
	If $a = c$, since $T[1] = c$ and $c \not \in \rest{\Pi}{S[p+1:n]}$, we have $\llrrbracket{T}[p+1] = |\rest{\Pi}{S[p+1:n]} \cup \{c\}| = |\rest{\Pi}{S[p+1:n]}| + 1$.
	If $\CR_S(c) = 0$ or $\CL_S(a) < \CL_S(c) \le \CR_S(c) < \CR_S(a)=p$,
	$\llrrbracket{S}[p]$ counts the number of distinct p-characters in $S[p+1:{}]S[{}:\CL_S(a)]$, where $c$ does not occur.
	On the other hand, $\llrrbracket{T}[p+1]$ counts the ones in $S[p+1:{}]cS[{}:\CL_S(a)]$.
	That is, $\llrrbracket{T}[p+1]=\llrrbracket{S}[p]+1$.
	Otherwise, if $\CL_S(c) < \CL_S(a)$ or $\CR_S(a) < \CR_S(c)$, we already have $c$ in $S[p+1:{}]S[{}:\CL_S(a)]$.
	Thus $\llrrbracket{T}[p+1] = \llrrbracket{S}[p]$.
	\qed
\end{proof}

\lempostatic*
\begin{proof}
	By definition, $\Pos = |T|_{\Sigma_{< c}} + |\{\, j \mid \First_{T}[j] = c,\, 1 \le j \le \Pos \,\}|$.
	By \Cref{cor:Kim4} and the bijectivity of $\LF_T$, the second term equals
	\[
	|\{\, i \mid \PBWT_{T}[i] = c,\, 1 \le i \le \LF_{T}^{-1}(\Pos) \,\}| 
	\]
	and further more equals
	\[
	|\{\, i \mid \PBWT^\circ_{T}[i] = c,\, 1 \le i \le \Dol \,\}| 
	\]
	because $\PBWT_{T}$ and $\PBWT^\circ_{T}$ are different only in that $\PBWT_{T}$ has an extra element $\$ < c$,
	and the position $\LF_{T}^{-1}(\Pos)$ in $\PBWT_T$ corresponds to the position $\Dol$ in $\PBWT^\circ_T$.
	\qed
\end{proof}

\lemequallcp*
\begin{proof}
	Let $S_p = U\$V$ and $S_q = XY\$Z$, where $|U\$| = |X| = p < q = |XY\$|$.
	Then, $T_p=U\$ cV$ and $T_q = XY \$ cZ$.
	Since $\$$ does not appear in $X$,
	\[
	\Lcp(\lrangle{S_p}, \lrangle{S_q}) = \Lcp(U\$,X) = \Lcp(\lrangle{T_p}, \lrangle{T_q})\,.
	\qedhere\]
\end{proof}

%% file: arxiv_pbwt.bbl
\begin{thebibliography}{10}

\bibitem{Baker93}
Brenda~S. Baker.
\newblock A theory of parameterized pattern matching: algorithms and
  applications.
\newblock In {\em Proceedings of the twenty-fifth annual ACM symposium on
  Theory of Computing (STOC 1993)}, pages 71--80, 1993.

\bibitem{PMA}
Brenda~S. Baker.
\newblock Parameterized pattern matching: Algorithms and applications.
\newblock {\em Journal of Computer and System Sciences}, 52(1):28--42, 1996.

\bibitem{Burrows94ablock-sorting}
Michael Burrows and David Wheeler.
\newblock A block-sorting lossless data compression algorithm.
\newblock Technical Report 124, Digital Equipment Corporation, 1994.

\bibitem{Diptarama2017}
Diptarama, Takashi Katsura, Yuhei Otomo, Kazuyuki Narisawa, and Ayumi
  Shinohara.
\newblock Position heaps for parameterized strings.
\newblock In {\em Proceedings of the 28th Annual Symposium on Combinatorial
  Pattern Matching (CPM 2017)}, pages 8:1--8:13, 2017.

\bibitem{EPM}
Kimmo Fredriksson and Maxim Mozgovoy.
\newblock Efficient parameterized string matching.
\newblock {\em Information Processing Letters}, 100(3):91 -- 96, 2006.

\bibitem{Fujisato2018}
Noriki Fujisato, Yuto Nakashima, Shunsuke Inenaga, Hideo Bannai, and Masayuki
  Takeda.
\newblock Right-to-left online construction of parameterized position heaps.
\newblock In {\em Proceedings of the Prague Stringology Conference 2018 (PSC
  2018)}, pages 91--102, 2018.

\bibitem{FujisatoSA2019}
Noriki Fujisato, Yuto Nakashima, Shunsuke Inenaga, Hideo Bannai, and Masayuki
  Takeda.
\newblock Direct linear time construction of parameterized suffix and {LCP}
  arrays for constant alphabets.
\newblock In {\em Proceedings of the 26th International Symposium on String
  Processing and Information Retrieval (SPIRE 2019)}, pages 382--391, 2019.

\bibitem{Ganguly2017}
Arnab Ganguly, Rahul Shah, and Sharma~V. Thankachan.
\newblock {pBWT}: Achieving succinct data structures for parameterized pattern
  matching and related problems.
\newblock In {\em Proceedings of the 28th Annual ACM-SIAM Symposium on Discrete
  Algorithms (SODA 2017)}, pages 397--407, 2017.

\bibitem{0002ST22}
Arnab Ganguly, Rahul Shah, and Sharma~V. Thankachan.
\newblock Fully functional parameterized suffix trees in compact space.
\newblock In {\em Proceedings of the 49th International Colloquium on Automata,
  Languages, and Programming (ICALP 2022)}, pages 65:1--65:18, 2022.

\bibitem{TomohiroDBIT09}
Tomohiro I, Satoshi Deguchi, Hideo Bannai, Shunsuke Inenaga, and Masayuki
  Takeda.
\newblock Lightweight parameterized suffix array construction.
\newblock In {\em Proceedings of the 20th International Workshop on
  Combinatorial Algorithms (IWOCA 2009)}, pages 312--323, 2009.

\bibitem{KimC21}
Sung{-}Hwan Kim and Hwan{-}Gue Cho.
\newblock Simpler {FM}-index for parameterized string matching.
\newblock {\em Information Processing Letters}, 165:106026, 2021.

\bibitem{DBLP:journals/dam/MendivelsoTP20}
Juan Mendivelso, Sharma~V. Thankachan, and Yoan~J. Pinz{\'{o}}n.
\newblock A brief history of parameterized matching problems.
\newblock {\em Discrete Applied Mathematics}, 274:103--115, 2020.

\bibitem{Nakashima2020}
Katsuhito Nakashima, Noriki Fujisato, Diptarama Hendrian, Yuto Nakashima, Ryo
  Yoshinaka, Shunsuke Inenaga, Hideo Bannai, Ayumi Shinohara, and Masayuki
  Takeda.
\newblock {DAWGs} for parameterized matching: Online construction and related
  indexing structures.
\newblock In {\em Proceedings of the 31st Annual Symposium on Combinatorial
  Pattern Matching (CPM 2020)}, pages 26:1--26:14, 2020.

\bibitem{Nakashima2020plst}
Katsuhito Nakashima, Diptarama Hendrian, Ryo Yoshinaka, and Ayumi Shinohara.
\newblock {An Extension of Linear-size Suffix Tries for Parameterized Strings}.
\newblock In {\em SOFSEM 2020 Student Research Forum}, pages 97--108, 2020.

\bibitem{Navarro2014}
Gonzalo Navarro and Yakov Nekrich.
\newblock Optimal dynamic sequence representations.
\newblock {\em SIAM Journal on Computing}, 43(5):1781--1806, 2014.

\bibitem{policriti15fast}
Alberto Policriti and Nicola Prezza.
\newblock Fast online {Lempel--Ziv} factorization in compressed space.
\newblock In {\em Proceedings of the the 22nd International Symposium on String
  Processing and Information Retrieval (SPIRE 2015)}, pages 13--20, 2015.

\bibitem{Shibuya04}
Tetsuo Shibuya.
\newblock Generalization of a suffix tree for {RNA} structural pattern
  matching.
\newblock {\em Algorithmica}, 39(1):1--19, 2004.

\bibitem{Thankachan22}
Sharma~V. Thankachan.
\newblock Compact text indexing for advanced pattern matching problems:
  Parameterized, order-isomorphic, {2D}, etc. (invited talk).
\newblock In {\em Proceedings of the 33rd Annual Symposium on Combinatorial
  Pattern Matching (CPM 2022)}, pages 3:1--3:3, 2022.

\end{thebibliography}
